\documentclass[a4paper,11pt]{article}

\usepackage[utf8]{inputenc}
\usepackage[UKenglish]{isodate}

\usepackage{mathrsfs}
\usepackage{bbm}
\usepackage{palatino}

\usepackage[margin=3cm]{geometry}

\usepackage[usenames,dvipsnames]{xcolor}
\usepackage{amsfonts}
\usepackage{amsmath}
\usepackage{amssymb}
\usepackage{amsthm}  
\usepackage{thmtools}
\usepackage{ifthen}
\usepackage{mathabx}
\usepackage{stmaryrd}
\usepackage[colorlinks=true,urlcolor=Mahogany,linkcolor=Mahogany,citecolor=Mahogany,plainpages=false,pdfpagelabels]{hyperref}
\hypersetup{pdftitle={Chain rules for quantum Rényi entropies}}
\usepackage{verbatim}
\usepackage[backend=bibtex8,style=alphabetic,isbn=false,maxcitenames=4,maxbibnames=100]{biblatex}
\usepackage{todonotes}

\newtheorem{theorem}{Theorem}
\newtheorem{lemma}[theorem]{Lemma}

\newtheorem{proposition}[theorem]{Proposition}
\newtheorem{definition}[theorem]{Definition}

\newtheorem{fact}[theorem]{Fact}

\newcommand{\sket}[1]{{\ensuremath{\lvert#1\rangle}}}
\newcommand{\lket}[1]{{\ensuremath{\left\lvert#1\right\rangle}}}
\newcommand{\ket}[1]{\mathchoice{\lket{#1}}{\sket{#1}}{\sket{#1}}{\sket{#1}}}

\newcommand{\sbra}[1]{{\ensuremath{\langle#1\rvert}}}
\newcommand{\lbra}[1]{{\ensuremath{\left\langle#1\right\rvert}}}
\newcommand{\bra}[1]{\mathchoice{\lbra{#1}}{\sbra{#1}}{\sbra{#1}}{\sbra{#1}}}

\newcommand{\sbraket}[2]{{\ensuremath{\langle#1\rvert#2\rangle}}}
\newcommand{\lbraket}[2]{{\ensuremath{\left\langle#1\!\left\rvert\vphantom{#1}#2\right.\!\right\rangle}}}
\newcommand{\braket}[2]{\mathchoice{\lbraket{#1}{#2}}{\sbraket{#1}{#2}}{\sbraket{#1}{#2}}{\sbraket{#1}{#2}}}

\newcommand{\sketbra}[2]{{\ensuremath{\lvert #1\rangle\langle #2\rvert}}}
\newcommand{\lketbra}[2]{{\ensuremath{\left\lvert #1 \middle\rangle\middle\langle #2\right\rvert}}}
\newcommand{\ketbra}[2]{\mathchoice{\lketbra{#1}{#2}}{\sketbra{#1}{#2}}{\sketbra{#1}{#2}}{\sketbra{#1}{#2}}}

\newcommand{\proj}[1]{\ketbra{#1}{#1}}

\newcommand{\ident}{\mathrm{id}}
\DeclareMathOperator{\tr}{\mathrm{Tr}}

\newcommand{\demi}{\frac{1}{2}}

\newcommand{\mbC}{\mathbb{C}}

\newcommand{\mbR}{\mathbb{R}}

\newcommand{\sfA}{\mathsf{A}}
\newcommand{\sfB}{\mathsf{B}}
\newcommand{\sfC}{\mathsf{C}}
\newcommand{\sfD}{\mathsf{D}}

\renewcommand{\otimes}{\varotimes}
\newcommand{\halpha}{\hat{\alpha}}
\DeclareMathOperator{\Op}{Op}
\DeclareMathOperator{\realpart}{Re}
\newcommand{\linop}{\mathrm{L}}
\newcommand{\Pos}{\mathrm{Pos}}
\newcommand{\D}{\mathrm{D}}
\DeclareMathOperator{\Supp}{supp}

\begin{document}
\cleanlookdateon
\title{Chain rules for quantum Rényi entropies}

\author{Frédéric Dupuis\footnote{Part of this work was performed while the author was at the Institute for Computer Science, Aarhus University, Aarhus, Denmark.}\\
{\it\small Faculty of Informatics}\\
{\it\small Masaryk University}\\
{\it\small Brno, Czech Republic}\\[2mm]
}


\maketitle

\begin{abstract}
    We present chain rules for a new definition of the quantum Rényi conditional entropy sometimes called the ``sandwiched'' Rényi conditional entropy. More precisely, we prove analogues of the equation $H(AB|C) = H(A|BC) + H(B|C)$, which holds as an identity for the von Neumann conditional entropy. In the case of the Rényi entropy, this relation no longer holds as an equality, but survives as an inequality of the form $H_{\alpha}(AB|C) \geqslant H_{\beta}(A|BC) + H_{\gamma}(B|C)$, where the parameters $\alpha, \beta, \gamma$ obey the relation $\frac{\alpha}{\alpha-1} = \frac{\beta}{\beta-1} + \frac{\gamma}{\gamma-1}$ and $(\alpha-1)(\beta-1)(\gamma-1) > 0$; if $(\alpha-1)(\beta-1)(\gamma-1) < 0$, the direction of the inequality is reversed.
\end{abstract}


\section{Introduction}
The Shannon entropy is one of the central concepts in information theory: it quantifies the amount of uncertainty contained in a random variable, and is used to characterize a wide range of information theoretical tasks. However, it is primarily useful for making asymptotic statements about problems involving repeated experiments, such as i.i.d.~sources or channels. If one is interested not only in asymptotic rates, but in how quickly we can approach these rates with increasing block sizes, the natural quantities that arise are Rényi entropies. Likewise, in the case of ``one-shot'' problems, such as those arising in cryptographic settings, the (smooth) min- and max-entropies are usually the relevant quantities. (For instance, these quantities are relevant in scenarios such as privacy amplification \cite{RK04,renner-phd}, decoupling \cite{fred-these,dbwr10}, and state merging \cite{diploma-berta,dbwr10}.)

All of these information theoretical problems have a quantum counterpart, and it is therefore desirable to construct quantum versions of these entropic quantities. To do this systematically, one can use the fact that each of these entropies can be defined via a notion of \emph{divergence} between two probability distributions; the entropy of a distribution $p$ is then derived from the divergence between $p$ and the uniform distribution. For example, the Shannon entropy can be defined in this way from the Kullback-Leibler divergence. One can then generalize the divergence quantities to quantum states and define quantum entropies in this way. However, these quantum counterparts to the divergences are not unique: the fact that quantum states can fail to commute leads to multiple non-equivalent definitions that all reduce to the classical version in the commutative case. In the case of the Rényi divergence, this phenomenon has led to multiple non-equivalent definitions being in use. One of them~\cite{op93} has been in use for several years and is defined as follows: let $\rho$ and $\sigma$ be two density operators. Then, their $\alpha$-Rényi-divergence is given by
\[ \tilde{D}_\alpha(\rho \| \sigma) := \frac{1}{\alpha-1} \log \frac{\tr[\rho^\alpha \sigma^{1-\alpha}]}{\tr[\rho]}. \]
This definition has found many applications, most notably in hypothesis testing~\cite{on00,quantum-chernoff,h06-2}, as well as for proving the fully quantum asymptotic equipartition property~\cite{tcr08,marco-these}. However, a new definition has recently been proposed in~\cite{mdsft13,wwy13} which seems to possess better properties than the traditional definition in certain parameter regimes. It is defined as:
\[ D_\alpha(\rho \| \sigma) := \frac{1}{\alpha-1} \log \left( \frac{1}{\tr[\rho]} \tr\left[ \left( \sigma^{\frac{1-\alpha}{2\alpha}} \rho \sigma^{\frac{1-\alpha}{2\alpha}} \right)^{\alpha} \right] \right). \]
(See Section~\ref{sec:renyi-definition} below for a more complete definition.) This new definition is sometimes called the ``sandwiched'' Rényi divergence, $\sigma$ being the bread and $\rho$ the meat in the expression above. The corresponding conditional Rényi entropy is then given by:
\[ H_{\alpha}(A|B)_{\rho} := -\min_{\sigma_B} D_{\alpha}(\rho_{AB} \| \ident_A \otimes \sigma_B) \]
for a bipartite state $\rho_{AB}$. This new definition has already found a large number of applications: it characterizes the strong converse exponent in hypothesis testing~\cite{mo13} (at least in some parameter regimes) and in classical-quantum channel coding~\cite{mo14}, and can be used to prove strong converses for a variety of information theoretical problems~\cite{wwy13,cmw14,tww14}. Furthermore, because of this wide applicability, the fundamental properties of $D_{\alpha}$ and $H_{\alpha}$ are being investigated: a number of properties have been proven in~\cite{mdsft13}, including the data processing inequality for $1 < \alpha \leqslant 2$ and a duality property (see Fact~\ref{fact:duality} below, also independently proven in \cite{b13}). The data processing inequality was also proven in~\cite{b13} for $\alpha > 1$ and in~\cite{fl13} for the full range $\alpha \geqslant \demi$. Moreover, in~\cite{ad13}, the quantity was generalized to so-called $\alpha$-$z$-divergences, and in~\cite{tbh14} the authors presented a duality property that involves both the new ``sandwiched'' definition and the traditional definition.

The present work is in this vein. One of the most fundamental properties of the von Neumann entropy is the so-called \emph{chain rule}: given a tripartite state $\rho_{ABC}$, we can break down the conditional entropy $H(AB|C)_{\rho}$ into two parts: $H(AB|C)_{\rho} = H(A|BC)_{\rho} + H(B|C)_{\rho}$. While this rule no longer holds as an equality for the Rényi entropy, one can nonetheless hope to salvage it in the form of inequalities, as was done in~\cite{dbwr10,vdtr12} for the smooth min-/max-entropies. The result is the main theorem of this paper:

\begin{theorem}\label{thm:main-result}
    Let $\rho_{ABC} \in \D(\sfA \otimes \sfB \otimes \sfC)$ be a normalized density operator, and let $\alpha, \beta, \gamma \in (\demi, 1) \cup (1, \infty)$ be such that $\frac{\alpha}{\alpha-1} = \frac{\beta}{\beta-1} + \frac{\gamma}{\gamma-1}$. Then, we have that
    \begin{equation}\label{eqn:first-case}
        H_{\alpha}(AB|C)_{\rho} \geqslant H_{\beta}(A|BC)_{\rho} + H_{\gamma}(B|C)_{\rho}
    \end{equation}
    if $(\alpha-1)(\beta-1)(\gamma-1) > 0$, and
    \begin{equation}\label{eqn:second-case}
        H_{\alpha}(AB|C)_{\rho} \leqslant H_{\beta}(A|BC)_{\rho} + H_{\gamma}(B|C)_{\rho}
    \end{equation}
    if $(\alpha-1)(\beta-1)(\gamma-1) < 0$.
\end{theorem}
The proof is given in Section~\ref{sec:main-result}, in the form of Propositions~\ref{prop:min-min-min} and~\ref{prop:min-min-max}, where the theorem statement is split into the case where all the signs are the same and the mixed-sign case, and where the statements proven are slightly more general. The proof is based on the Riesz-Thorin-type norm interpolation techniques developed in \cite{b13}, which are then applied to convenient expressions for the Rényi entropy.


\section{Preliminaries}

\subsection{Notation}\label{sec:notation}
In the table below, we summarize the notation used throughout the paper:
\begin{center}
    \begin{tabular}{|c|l|}
        \hline
        \emph{Symbol} & \multicolumn{1}{c|}{\emph{Definition}}\\
        \hline
        $A, B, C, \dots$ & Quantum systems\\
        $\sfA,\sfB,\dots$ & Hilbert spaces corresponding to systems $A, B,\dots$\\
        $\linop(\sfA, \sfB)$ & Set of linear operators from $\sfA$ to $\sfB$\\
        $\linop(\sfA)$ & $\linop(\sfA, \sfA)$\\
        $X_{AB}$ & Operator in $\linop(\sfA \otimes \sfB)$\\
        $X_{A \rightarrow B}$ & Operator in $\linop(\sfA, \sfB)$\\
        $\Pos(\sfA)$ & Set of positive semidefinite operators on $\sfA$\\
        $\D(\sfA)$ & Set of positive semidefinite operators on $\sfA$ with unit trace\\
        $\ident_A$ & Identity operator on $\sfA$\\
        $\alpha', \beta', \gamma'$ & $\frac{\alpha-1}{\alpha}, \frac{\beta-1}{\beta}, \frac{\gamma-1}{\gamma}$\\
        $\hat{\alpha}, \hat{\beta}, \hat{\gamma}$ & Dual parameter of $\alpha, \beta, \gamma$: $\frac{1}{\alpha} + \frac{1}{\hat{\alpha}} = 2$, etc.\\
        $\Op_{A \rightarrow B}(\ket{i}_A \otimes \ket{i}_B)$ & $\ket{j}_B \bra{i}_A$, for computational basis vectors $\ket{i}_A, \ket{j}_B$  \\
        $X^{\dagger}$ & Adjoint of $X$.\\
        $X^{\top}$ & Transpose of $X$ with respect to the computational basis.\\
        $X_{A} \geqslant Y_A$ & $X-Y \in \Pos(\sfA)$.\\
        $\| X \|_p$ & $\tr[(X^{\dagger} X)^{\frac{p}{2}}]^{\frac{1}{p}}$. Note that if $p<1$, this is not a norm.\\
        $\Supp X$ & Support of $X$\\
        \hline
    \end{tabular}
\end{center}
Note in particular the use of the shorthand $\alpha'$ to denote $\frac{\alpha-1}{\alpha}$, i.e.~the inverse of the Hölder conjugate of $\alpha$; it will be used extensively. Using this shorthand, $\alpha \in (\demi, 1)$ corresponds to $\alpha' \in (-1,0)$ and $\alpha > 1$ corresponds to $\alpha' \in (0,1)$.

Another convention used extensively throughout the paper is the implicit tensorization by the identity: if we have two operators $X_{AB}$ and $Y_{B}$, by $X_{AB} Y_B$ we mean $X_{AB}(\ident_A \otimes Y_B)$. We will also often omit subscripts when doing so should not create confusion.

We will also make use of the operator-vector correspondence. We will endow every Hilbert space with its own computational basis, denoted by $\ket{i}_A$ for the space $\sfA$ and so on, and we define the linear map $\Op_{A \rightarrow B}: \sfA \otimes \sfB \rightarrow \linop(\sfA, \sfB)$ by its action on the computational basis as follows: $\Op_{A \rightarrow B}(\ket{i}_A \otimes \ket{j}_B) = \ket{j}_B \bra{i}_A$.

\subsection{The quantum Rényi entropy}\label{sec:renyi-definition}
We now present the definition of the quantum Rényi divergence, first defined in~\cite{mdsft13,wwy13} and sometimes called the ``sandwiched'' Rényi divergence:
\begin{definition}[Quantum Rényi divergence]
    Let $\rho, \sigma \in \Pos(\sfA)$, and let $\alpha \in [\demi, 1) \cup (1, \infty)$. Then, we define their Rényi $\alpha$-divergence as
        \begin{align*}
            D_{\alpha}(\rho \| \sigma) &:= \frac{1}{\alpha-1} \log\left( \frac{1}{\tr[\rho]} \tr\left[ \left( \sigma^{\frac{-\alpha'}{2}} \rho \sigma^{\frac{-\alpha'}{2}} \right)^{\alpha} \right] \right),
        \end{align*}
        unless $\alpha > 1$ and $\Supp \rho \nsubseteq \Supp \sigma$, in which case $D_{\alpha}(\rho \| \sigma) := \infty$.
\end{definition}

The Rényi entropy is then defined as follows:
\begin{definition}[Quantum Rényi entropy]
    Let $\rho_{AB} \in \Pos(\sfA \otimes \sfB)$ and $\sigma_B \in \D(\sfB)$, and let $\alpha \in [\demi, 1) \cup (1, \infty)$. Then,
        \[ H_{\alpha}(A|B)_{\rho|\sigma} := -D_{\alpha}(\rho_{AB} \| \ident_A \otimes \sigma_B), \]
    and
        \[ H_{\alpha}(A|B)_{\rho} := -\inf_{\omega_B \in \D(\sfB)} D_{\alpha}(\rho_{AB} \| \ident_A \otimes \omega_B). \]
\end{definition}
By taking the limit as $\alpha \rightarrow 1$, we recover the von Neumann entropy; by taking the limit $\alpha \rightarrow \infty$, we get the min-entropy~\cite{renner-phd}; by choosing $\alpha=\demi$, we get the max-entropy \cite{min-max-entropy}; and by choosing $\alpha=2$ we get the collision entropy from \cite{renner-phd}.

We also recall the \emph{duality} property of quantum Rényi entropies proven in~\cite[Theorem 9]{mdsft13} and also independently in \cite[Theorem 9]{b13}, which is the Rényi analogue of the duality between min- and max-entropy:
\begin{fact}[Duality of Rényi entropies]\label{fact:duality}
    Let $\ket{\psi}_{ABC} \in \sfA \otimes \sfB \otimes \sfC$ be a normalized pure state, and let $\rho_{ABC} := \proj{\psi}_{ABC}$. Then, we have that for any $\alpha \in (\demi, 1) \cup (1, \infty)$,
    \[ H_{\alpha}(A|B)_{\rho} = -H_{\hat{\alpha}}(A|C)_{\rho}, \]
    where $\frac{1}{\alpha} + \frac{1}{\hat{\alpha}} = 2$.
\end{fact}
Note here that it is particularly convenient to rephrase the condition $\frac{1}{\alpha} + \frac{1}{\hat{\alpha}} = 2$ using the shorthand given in the table in Section~\ref{sec:notation}: it corresponds to $\alpha' = -\hat{\alpha}'$.


\section{Proof of the main result}\label{sec:main-result}

The first ingredient of the proof is the following generalization of Hadamard's three-line theorem proven by Beigi in~\cite{b13} as part of his proof of the data processing inequality of the sandwiched Rényi divergence for $\alpha > 1$:
\begin{theorem}[Theorem 2 from~\cite{b13}]\label{thm:interpolation}
    Let $F:S \rightarrow \linop(\sfA)$, where $S := \{ z \in \mbC : 0 \leqslant \realpart(z) \leqslant 1 \}$, be a bounded map that is holomorphic on the interior of $S$ and continuous on the boundary. Let $0 < \theta < 1$ and define $p_{\theta}$ by
    \[ \frac{1}{p_{\theta}} = \frac{1 - \theta}{p_0} + \frac{\theta}{p_1}. \]
    For $k = 0,1$ define
    \[ M_k = \sup_{t \in \mbR} \| F(k+it) \|_{p_k}. \]
    Then, we have
    \[ \| F(\theta) \|_{p_{\theta}} \leqslant M_0^{1-\theta} M_1^{\theta}. \]
\end{theorem}
Note that in~\cite{b13}, the theorem is stated for more general norms involving a positive operator $\sigma$ (which can be chosen to be $\sigma = \ident$ to obtain this version) and adds a condition that $p_0 \leqslant p_1$ which is not necessary, as can be seen by applying the theorem with $\bar{F}(z) := F(1-z)$ and $\bar{\theta} = 1-\theta$.

The second ingredient is the following lemma, which gives a particularly useful expression for the Rényi entropy:
\begin{lemma}\label{lem:renyi-expressions}
    Let $\ket{\psi}_{ABCD} \in \sfA \otimes \sfB \otimes \sfC \otimes \sfD$ be a normalized pure state, let $X_{AD \rightarrow BC} = \Op_{AD \rightarrow BC}(\ket{\psi})$, and let $\alpha \in (\demi, 1) \cup (1, \infty)$, with $\halpha$ such that $\frac{1}{\alpha} + \frac{1}{\halpha} = 2$, and let $\rho_{ABCD} = \proj{\psi}$ and $\sigma_C \in \D(\sfC)$. Then,
    \begin{align}
        \label{eqn:halpha-rho-tau} H_{\alpha}(AB|C)_{\rho|\sigma} &= -\log \sup_{\tau_D \in \D(\sfD)} \left\| \sigma_C^{\frac{-\alpha'}{2}} X \tau_D^{\frac{\alpha'}{2}} \right\|^{\frac{2}{\alpha'}}_2,\\
        \label{eqn:halpha-rho} H_{\alpha}(B|C)_{\rho|\sigma} &= -\log \left\| \sigma_{C}^{\frac{-\alpha'}{2}} X \right\|^{\frac{2}{\alpha'}}_{2\alpha},\\
        \label{eqn:halpha-tau} H_{\alpha}(A|BC)_{\rho} &= -\log \sup_{\tau_D \in \D(\sfD)} \left\| X  {\tau_D}^{\frac{\alpha'}{2}} \right\|^{\frac{2}{\alpha'}}_{2\hat{\alpha}}.
    \end{align}
\end{lemma}
\begin{proof}
    We start by proving (\ref{eqn:halpha-rho-tau}) from the represention in equation (19) from~\cite{mdsft13}:
    \[ H_{\alpha}(AB|C)_{\rho|\sigma} = \begin{cases}
            \frac{-1}{\alpha'} \log \inf_{\tau_D \in \D(\sfD)} \bra{\psi} \ident_{AB} \otimes \sigma_C^{-\alpha'} \otimes \tau_D^{\alpha'} \ket{\psi} & \text{if $\alpha<1$}\\
            \frac{-1}{\alpha'} \log \sup_{\tau_D \in \D(\sfD)} \bra{\psi} \ident_{AB} \otimes \sigma_C^{-\alpha'} \otimes \tau_D^{\alpha'} \ket{\psi} & \text{if $\alpha>1$},
    \end{cases} \]
    Putting the prefactor in the exponent, we get
    \[ H_{\alpha}(AB|C)_{\rho|\sigma} = -\log \sup_{\tau_D \in \D(\sfD)} \bra{\psi} \ident_{AB} \otimes \sigma_C^{-\alpha'} \otimes \tau_D^{\alpha'} \ket{\psi}^{\frac{1}{\alpha'}} \]
    for all $\alpha \in [\demi, 1) \cup (1, \infty)$.  Now, consider the vector $\sigma_C^{\frac{-\alpha'}{2}} \otimes \tau_D^{\frac{\alpha'}{2}} \ket{\psi}$: by Lemma~\ref{lem:op-vec-xy}, we have that
    \[ \Op_{AD \rightarrow BC}\left(\sigma_C^{\frac{-\alpha'}{2}} \otimes \tau_D^{\frac{\alpha'}{2}} \ket{\psi}\right) = \sigma_C^{\frac{-\alpha'}{2}} X {\tau_D^{\top}}^{\frac{\alpha'}{2}} \]
    and therefore
    \begin{align*}
        \bra{\psi} \ident_{AB} \otimes {\sigma_C}^{-\alpha'} \otimes \tau_D^{\alpha'} \ket{\psi}^{\frac{1}{\alpha'}} &= \left\| \Op_{AD \rightarrow BC}\left(\sigma_C^{\frac{-\alpha'}{2}} \otimes {\tau_D^{\top}}^{\frac{\alpha'}{2}} \ket{\psi}\right) \right\|_2^{\frac{2}{\alpha'}}\\
        &= \left\| \sigma_C^{\frac{-\alpha'}{2}} X {\tau_D^{\top}}^{\frac{\alpha'}{2}} \right\|_2^{\frac{2}{\alpha'}}
    \end{align*}
    by Lemma~\ref{lem:op-vec-norm}, from which~\eqref{eqn:halpha-rho-tau} follows. We now derive~\eqref{eqn:halpha-rho} by expressing $H_{\alpha}(B|C)_{\rho|\sigma}$ using~\eqref{eqn:halpha-rho-tau}:
    \begin{align*}
        2^{-H_{\alpha}(B|C)_{\rho|\sigma}} &= \sup_{\tau_{AD} \in \D(\sfA \otimes \sfD)} \left\| \sigma_{C}^{\frac{-\alpha'}{2}} X {\tau_{AD}}^{\frac{\alpha'}{2}} \right\|_2^{\frac{2}{\alpha'}}\\
        &= \sup_{\tau_{AD} \in \D(\sfA \otimes \sfD)} \tr\left[ X^{\dagger} \sigma_{C}^{-\alpha'} X {\tau_{AD}}^{\alpha'} \right]^{\frac{1}{\alpha'}}\\
        &= \left\| X^{\dagger} {\sigma_C}^{-\alpha'} X \right\|_{\alpha}^{\frac{1}{\alpha'}}\\
        &= \left\| {\sigma_C}^{\frac{-\alpha'}{2}} X \right\|^{\frac{2}{\alpha'}}_{2\alpha},
    \end{align*}
    where the third line follows from Lemma~\ref{lem:supY}, and~\eqref{eqn:halpha-rho} follows. Finally, we prove~\eqref{eqn:halpha-tau} via duality: 
    \begin{align*}
        H_{\alpha}(A|BC)_{\rho} &= -H_{\hat{\alpha}}(A|D)_{\rho}\\
        &= \log \inf_{\omega_D \in \D(\sfD)} \left\| X  {\omega_D}^{\frac{-\halpha'}{2}} \right\|_{2\halpha}^{\frac{2}{\halpha'}}\\
        &= \log \inf_{\omega_D \in \D(\sfD)} \left\| X  {\omega_D}^{\frac{\alpha'}{2}} \right\|_{2\halpha}^{\frac{-2}{\alpha'}}\\
        &= -\log \sup_{\omega_D \in \D(\sfD)} \left\| X  {\omega_D}^{\frac{\alpha'}{2}} \right\|_{2\halpha}^{\frac{2}{\alpha'}},
    \end{align*}
    where we used the fact that $\halpha' = -\alpha'$, and invoked Lemma~\ref{lem:aux-expr-hab} to get the third line. This concludes the proof.
\end{proof}

We are now ready to prove Theorem~\ref{thm:main-result}. We break it down into its two cases, with Proposition~\ref{prop:min-min-min} below corresponding to Equation~(\ref{eqn:first-case}) and Proposition~\ref{prop:min-min-max} further down corresponding to Equation~(\ref{eqn:second-case}). Note that we only prove the cases where $\alpha > 1$; the cases where $\alpha < 1$ then follow by applying the duality property (Fact \ref{fact:duality}) to all the terms.
\begin{proposition}[Theorem~\ref{thm:main-result}, all-same-sign case]\label{prop:min-min-min}
    Let $\rho_{ABC} \in \D(\sfA \otimes \sfB \otimes \sfC)$, and let $\alpha,\beta,\gamma > 1$ be such that $\frac{1}{\alpha'} = \frac{1}{\beta'} + \frac{1}{\gamma'}$. Then, we have that
    \[ H_{\alpha}(AB|C)_{\rho|\sigma} \geqslant H_{\beta}(A|BC)_{\rho} + H_{\gamma}(B|C)_{\rho|\sigma} \]
for any $\sigma_C \in \D(\sfC)$. In particular,
    \[ H_{\alpha}(AB|C)_{\rho} \geqslant H_{\beta}(A|BC)_{\rho} + H_{\gamma}(B|C)_{\rho}. \]
\end{proposition}
\begin{proof}
    Let $\ket{\psi}_{ABCD} \in \sfA \otimes \sfB \otimes \sfC \otimes \sfD$ be a purification of $\rho$, let $X_{AD \rightarrow BC} = \Op_{AD \rightarrow BC}(\ket{\psi})$, and let $\tau_D \in \D(\sfD)$. We use Theorem~\ref{thm:interpolation} with the following choices:
    \begin{align*}
        F(z) &= \sigma_C^{\frac{-z \gamma'}{2}} X \tau_D^{\frac{(1-z) \beta'}{2}}\\
        \frac{1}{p_0} &= 1 - \frac{1}{2\beta} = \frac{1}{2\hat{\beta}}\\
        p_1 &= 2\gamma\\
        \theta &= \frac{\alpha'}{\gamma'}.
    \end{align*}
    Now, from these choices, we can show that $1-\theta = \frac{\alpha'}{\beta'}$:
    \begin{align*}
        \frac{1-\theta}{\alpha'} &= \frac{1 - \alpha'/\gamma'}{\alpha'}\\
            &= \frac{1}{\alpha'} - \frac{1}{\gamma'}\\
            &= \frac{1}{\beta'}.
    \end{align*}
    Furthermore, we can show that $p_{\theta} = 2$:
    \begin{align}
        \label{eqn:ptheta} \frac{1}{p_{\theta}} &= \frac{1-\theta}{p_0} + \frac{\theta}{p_1}\\
        \nonumber &= \frac{\alpha'}{2 \hat{\beta} \beta'} + \frac{\alpha'}{2 \gamma' \gamma}\\
        \nonumber &= \frac{\alpha'}{2}\left( \frac{1}{\hat{\beta} \beta'} + \frac{1}{\gamma' \gamma} \right)\\
        \nonumber &\stackrel{(a)}{=} \frac{\alpha'}{2}\left( \frac{1 + \beta'}{\beta'} + \frac{1 - \gamma'}{\gamma'} \right)\\
        \nonumber &= \frac{\alpha'}{2}\left( \frac{1}{\beta'} + \frac{1}{\gamma'} \right)\\
        \nonumber &= \frac{1}{2},
    \end{align}
    where at line $(a)$, we have used the fact that $\frac{1}{\gamma} = 1 - \gamma'$, and that $\frac{1}{\hat{\beta}} = 2 - \frac{1}{\beta} = 1 + \beta'$. Finally, also note that $F(\theta) = \sigma_C^{\frac{-\alpha'}{2}} X \tau_D^{\frac{\alpha'}{2}}$, and that
    \begin{align*}
    M_0 &= \sup_{t \in \mbR} \left\| \sigma_C^{\frac{-it \gamma'}{2}} X \tau_D^{\frac{(1-it)\beta'}{2}} \right\|_{p_0}\\
    &= \left\| X \tau_D^{\frac{\beta'}{2}} \right\|_{p_0}
    \end{align*}
    due to the fact that $\sigma_C^{\frac{-it\gamma'}{2}}$ and $\tau_D^{\frac{-it\beta'}{2}}$ are unitary for all $t \in \mbR$. Likewise, 
    \[ M_1 = \left\| \sigma_C^{\frac{-\gamma'}{2}} X \right\|_{p_1}. \]
    Hence, for any choice of normalized $\sigma$ and $\tau$, we have from Theorem \ref{thm:interpolation} that
    \[ \left\| \sigma_C^{\frac{-\alpha'}{2}} X \tau_D^{\frac{\alpha'}{2}} \right\|_2 \leqslant \left\| X \tau_D^{\frac{\beta'}{2}} \right\|_{2\hat{\beta}}^{\frac{\alpha'}{\beta'}} \left\| \sigma_C^{\frac{-\gamma'}{2}} X \right\|_{2\gamma}^{\frac{\alpha'}{\gamma'}}. \]
    This leads to
    \begin{align*}
    \left\| \sigma_C^{\frac{-\alpha'}{2}} X \tau_D^{\frac{\alpha'}{2}} \right\|_2^{\frac{2}{\alpha'}} &\leqslant \left\| X \tau_D^{\frac{\beta'}{2}} \right\|_{2\hat{\beta}}^{\frac{2}{\beta'}} \left\| \sigma_C^{\frac{-\gamma'}{2}} X \right\|_{2\gamma}^{\frac{2}{\gamma'}},
    \end{align*}
    since $\alpha' > 0$. Maximizing over $\tau_D$ on both sides yields:
    \begin{align*}
    \sup_{\tau_D \in \D(\sfD)} \left\| \sigma_C^{\frac{-\alpha'}{2}} X \tau_D^{\frac{\alpha'}{2}} \right\|_2^{\frac{2}{\alpha'}} &\leqslant \sup_{\tau_D \in \D(\sfD)} \left\| X \tau_D^{\frac{\beta'}{2}} \right\|_{2\hat{\beta}}^{\frac{2}{\beta'}} \left\| \sigma_C^{\frac{-\gamma'}{2}} X \right\|_{2\gamma}^{\frac{2}{\gamma'}}.
    \end{align*}
    Finally, using Lemma~\ref{lem:renyi-expressions}, we get that
    \[ H_{\alpha}(AB|C)_{\rho|\sigma} \geqslant H_{\beta}(A|BC)_{\rho} + H_{\gamma}(B|C)_{\rho|\sigma}, \]
    as advertised.
\end{proof}

We now turn to the second case of Theorem~\ref{thm:main-result}, corresponding to Equation~(\ref{eqn:second-case}), using essentially the same proof:
\begin{proposition}[Theorem~\ref{thm:main-result}, mixed-sign case]\label{prop:min-min-max}
    Let $\rho_{ABC} \in \D(\sfA \otimes \sfB \otimes \sfC)$, and let $\alpha'>0$, and $\beta'$ and $\gamma'$ have opposite signs, and be such that $\frac{1}{\alpha'} = \frac{1}{\beta'} + \frac{1}{\gamma'}$. Then, we have that
    \[ H_{\alpha}(AB|C)_{\rho|\sigma} \leqslant H_{\beta}(A|BC)_{\rho} + H_{\gamma}(B|C)_{\rho|\sigma} \]
    for every $\sigma_C \in \D(\sfC)$. In particular,
    \[ H_{\alpha}(AB|C)_{\rho} \leqslant H_{\beta}(A|BC)_{\rho} + H_{\gamma}(B|C)_{\rho}. \]
\end{proposition}
\begin{proof}
    Let $\ket{\psi}_{ABCD} \in \sfA \otimes \sfB \otimes \sfC \otimes \sfD$ be a purification of $\rho$, let $X_{AD \rightarrow BC} = \Op_{AD \rightarrow BC}(\ket{\psi})$, and let $\tau_D \in \D(\sfD)$. We split the proof into two cases: first, we assume that $\beta'>0$. We use Theorem~\ref{thm:interpolation} with the following choices:
    \begin{align*}
        F(z) &= \sigma_C^{\frac{-\alpha'}{2} - z \frac{\alpha' \gamma'}{2\beta'}} X \tau_D^{\frac{\alpha'}{2} - z \frac{\alpha'}{2}}\\
        p_0 &= 2\\
        p_1 &= 2\gamma\\
        \theta &= \frac{-\beta'}{\gamma'}
    \end{align*}
    Now, note that $F(\theta) = X \tau_D^{\frac{\beta'}{2}}$, and that $1-\theta = \frac{\beta'}{\alpha'}$ and $p_{\theta} = 2\hat\beta$, as can be seen by a calculation very similar to Equation~(\ref{eqn:ptheta}). Furthermore, we have that:
    \begin{align*}
        M_0 &= \sup_{t \in \mbR} \left\| \sigma_C^{\frac{-\alpha'}{2} - it \frac{\alpha'\gamma'}{2\beta'}} X \tau_D^{\frac{\alpha'}{2} - it \frac{\alpha'}{2}} \right\|_{2}\\
        &= \left\| \sigma_C^{\frac{-\alpha'}{2}}X \tau_D^{\frac{\alpha'}{2}} \right\|_{2}
    \end{align*}
    due to the fact that $\sigma_C^{\frac{-it\alpha'\gamma'}{2\beta'}}$ and $\tau_D^{\frac{-it\alpha'}{2}}$ are unitary for all $t \in \mbR$. Likewise, 
    \[ M_1 = \left\| \sigma_C^{\frac{-\gamma'}{2}} X \right\|_{2\gamma}. \]
    Hence, for any choice of normalized $\sigma$ and $\tau$, we have that
    \[ \left\| X \tau_D^{\frac{\beta'}{2}} \right\|_{2\hat\beta} \leqslant \left\| \sigma_C^{\frac{-\alpha'}{2}} X \tau_D^{\frac{\alpha'}{2}} \right\|_{2}^{\frac{\beta'}{\alpha'}} \left\| \sigma_C^{\frac{-\gamma'}{2}} X \right\|_{2\gamma}^{\frac{-\beta'}{\gamma'}}. \]
    Now, since $\beta' > 0$, this leads to:
    \[ \left\| X \tau_D^{\frac{\beta'}{2}} \right\|_{2\hat\beta}^{\frac{2}{\beta'}} \leqslant \left\| \sigma_C^{\frac{-\alpha'}{2}} X \tau_D^{\frac{\alpha'}{2}} \right\|_{2}^{\frac{2}{\alpha'}} \left\| \sigma_C^{\frac{-\gamma'}{2}} X \right\|_{2\gamma}^{\frac{-2}{\gamma'}}. \]
    and therefore
    \begin{align*}
        \sup_{\tau_D \in \D(\sfD)} \left\| X \tau_D^{\frac{\beta'}{2}} \right\|_{2\hat\beta}^{\frac{2}{\beta'}} &\leqslant \sup_{\tau_D \in \D(\sfD)} \left\| \sigma_C^{\frac{-\alpha'}{2}} X \tau_D^{\frac{\alpha'}{2}} \right\|_{2}^{\frac{2}{\alpha'}} \left\| \sigma_C^{\frac{-\gamma'}{2}} X \right\|_{2\gamma}^{\frac{-2}{\gamma'}}.
    \end{align*}
    Using Lemma~\ref{lem:renyi-expressions}, we get that
    \[ H_{\beta}(A|BC)_{\rho} \geqslant H_{\alpha}(AB|C)_{\rho|\sigma} - H_{\gamma}(B|C)_{\rho|\sigma}, \]
    or,
    \[ H_{\alpha}(AB|C)_{\rho|\sigma} \leqslant H_{\beta}(A|BC)_{\rho} + H_{\gamma}(B|C)_{\rho|\sigma}. \]

    We now turn to the case where $\beta' < 0$ (and therefore $\gamma' > 0$). We again use Theorem~\ref{thm:interpolation}, but with these choices:
    \begin{align*}
        F(z) &= \sigma_C^{\frac{-\alpha'}{2} + z \frac{\alpha'}{2}} X \tau_C^{\frac{\alpha'}{2} + z \frac{\alpha'\beta'}{2\gamma'}}\\
        p_0 &= 2\\
        p_1 &= 2\hat\beta\\
        \theta &= \frac{-\gamma'}{\beta'}
    \end{align*}
    Now, note that $F(\theta) = \sigma_C^{\frac{-\gamma'}{2}} X$, that $1-\theta = \frac{\gamma'}{\alpha'}$, and that $p_{\theta} = 2\gamma$ (see again Equation (\ref{eqn:ptheta}) for a similar calculation), and that
    \begin{align*}
        M_0 &= \sup_{t \in \mbR} \left\| \sigma_C^{\frac{-\alpha'}{2} - it \frac{\alpha'}{2}} X \tau_D^{\frac{\alpha'}{2} - it \frac{\alpha'\beta'}{2\gamma'}} \right\|_{2}\\
        &= \left\| \sigma_C^{\frac{-\alpha'}{2}}X \tau_D^{\frac{\alpha'}{2}} \right\|_{2}
    \end{align*}
    due to the fact that $\sigma_C^{\frac{-it\alpha'}{2}}$ and $\tau_D^{\frac{-it\alpha'\beta'}{2\gamma'}}$ are unitary for all $t \in \mbR$. Likewise, 
    \[ M_1 = \left\| X \tau_D^{\frac{\beta'}{2}} \right\|_{2\hat\beta}. \]
    Hence, for any choice of normalized $\sigma$ and $\tau$, we have that
    \[ \left\| \sigma_C^{\frac{-\gamma'}{2}} X  \right\|_{2\gamma} \leqslant \left\| \sigma_C^{\frac{-\alpha'}{2}} X \tau_D^{\frac{\alpha'}{2}} \right\|_{2}^{\frac{\gamma'}{\alpha'}} \left\|  X \tau_D^{\frac{\beta'}{2}} \right\|_{2\hat\beta}^{\frac{-\gamma'}{\beta'}}. \]
    Since $\gamma' > 0$, this leads to
    \[ \left\| \sigma_C^{\frac{-\gamma'}{2}} X  \right\|_{2\gamma}^{\frac{2}{\gamma'}} \leqslant \left\| \sigma_C^{\frac{-\alpha'}{2}} X \tau_D^{\frac{\alpha'}{2}} \right\|_{2}^{\frac{2}{\alpha'}} \left\|  X \tau_D^{\frac{\beta'}{2}} \right\|_{2\hat\beta}^{\frac{-2}{\beta'}}. \]
    Moving the rightmost term to the left-hand side and maximizing over $\tau_D$ on both sides, we get:
    \begin{align*}
        \left\| \sigma_C^{\frac{-\gamma'}{2}} X  \right\|_{2\gamma}^{\frac{2}{\gamma'}} \left( \sup_{\tau_D \in \D(\sfD)} \left\|  X \tau_D^{\frac{\beta'}{2}} \right\|_{2\hat\beta}^{\frac{2}{\beta'}} \right) &\leqslant \sup_{\tau_D \in \D(\sfD)} \left\| \sigma_C^{\frac{-\alpha'}{2}} X \tau_D^{\frac{\alpha'}{2}} \right\|_{2}^{\frac{2}{\alpha'}}.
    \end{align*}
    Using Lemma~\ref{lem:renyi-expressions}, we get that
    \[ H_{\gamma}(B|C)_{\rho|\sigma} + H_{\beta}(A|BC)_{\rho} \geqslant H_{\alpha}(AB|C)_{\rho|\sigma}. \]
    This concludes the proof.
\end{proof}


\section*{Acknowledgments}
The author would like to thank Marco Tomamichel, Salman Beigi, Andreas Winter, Ciara Morgan, and Renato Renner for fruitful discussions and comments. The author acknowledges the support of the Czech Science Foundation GA CR project P202/12/1142 and the support of the EU FP7 under grant agreement no 323970 (RAQUEL). Part of this work was performed while the author was at Aarhus University, where he was supported by the Danish National Research Foundation and The National Science Foundation of China (under the grant 61061130540) for the Sino-Danish Center for the Theory of Interactive Computation, and also by the CFEM research centre (supported by the Danish Strategic Research Council).

\appendix

\section{Auxilliary lemmas}

\begin{lemma}\label{lem:aux-expr-hab}
    Let $\rho_{AB} \in \D(\sfA \otimes \sfB)$. Then, we have that
    \[ H_{\alpha}(A|B) = -\log \inf_{\sigma_B \in \D(\sfB)} \left\| X \sigma_{B}^{\frac{-\alpha'}{2}} \right\|_{2\alpha}^{\frac{2}{\alpha'}}, \]
    where $X := \Op_{AB \rightarrow D}(\ket{\psi})$ for a purification $\ket{\psi}_{ABD}$ of $\rho$.
\end{lemma}
\begin{proof}
    First, we write
    \begin{align*}
        H_{\alpha}(A|B)_{\rho|\sigma} &= -D_{\alpha}(\rho_{AB} \| \ident_A \otimes \sigma_B)\\
        &= \frac{1}{1-\alpha} \log \tr\left[ \left( \sigma_B^{\frac{-\alpha'}{2}} \rho \sigma_B^{\frac{-\alpha'}{2}} \right)^{\alpha} \right]\\
        &= \frac{-1}{\alpha'} \log \left\| \sigma_B^{\frac{-\alpha'}{2}} \rho \sigma_B^{\frac{-\alpha'}{2}} \right\|_{\alpha}\\
        &= -\log \left\| \sigma_B^{\frac{-\alpha'}{2}} X^{\dagger} X \sigma_B^{\frac{-\alpha'}{2}} \right\|^{\frac{1}{\alpha'}}_{\alpha}\\
        &= -\log \left\| X \sigma_B^{\frac{-\alpha'}{2}} \right\|^{\frac{2}{\alpha'}}_{2\alpha}\\
    \end{align*}
    where we have used the fact that $\rho_{AB} = X^{\dagger} X$.
\end{proof}

\begin{lemma}\label{lem:supY}
    Let $\alpha \in [\demi, 1) \cup (1, \infty)$. Then, for any $X \in \Pos(\sfA)$, we have that
        \[ \| X \|_{\alpha} = \sup_{Y \in \D(\sfA)} \tr[Y^{\alpha'} X] \]
        if $\alpha > 1$, and
        \[ \| X \|_{\alpha} = \inf_{Y \in \D(\sfA)} \tr[Y^{\alpha'} X] \]
        if $\alpha < 1$. (As usual, $\alpha' = \frac{\alpha-1}{\alpha}$.) In particular, this means that
        \[ \| X \|_{\alpha}^{\frac{1}{\alpha'}} = \sup_{Y \in \D(\sfA)} \tr[Y^{\alpha'} X]^{\frac{1}{\alpha'}}. \]
\end{lemma}
\begin{proof}
    This is simply a reformulation of Lemma 11 in \cite{mdsft13}.
\end{proof}

\begin{lemma}\label{lem:op-vec-xy}
    Let $\ket{\psi} \in \sfA \otimes \sfB$, and let $X_A \in \linop(\sfA)$ and $Y_B \in \linop(\sfB)$. Then, we have that
    \[ \Op_{A \rightarrow B}(X_A \otimes Y_B \ket{\psi}) = Y_B \Op(\ket{\psi}) X_A^{\top}. \]
\end{lemma}
\begin{proof}
    This can be shown by a simple manipulation of indices; see for example \cite[Section 2.4]{watrousnotes11}.
\end{proof}

\begin{lemma}\label{lem:op-vec-norm}
    Let $\ket{\psi} \in \sfA \otimes \sfB$. Then,
    \[ \braket{\psi}{\psi} = \tr[\Op_{A \rightarrow B}(\ket{\psi})^{\dagger} \Op_{A \rightarrow B}(\ket{\psi})] \]
    and therefore,
    \[ \| \ket{\psi} \| = \| \Op_{A \rightarrow B}(\ket{\psi}) \|_2.  \]
\end{lemma}
\begin{proof}
    Again, see \cite[Section 2.4]{watrousnotes11}.
\end{proof}

\printbibliography


\end{document}